\documentclass{article}
\usepackage[margin=1.4in]{geometry}
\usepackage[utf8]{inputenc}

\usepackage{amsmath}
\usepackage{amssymb}
\usepackage{amsthm}
\usepackage{amstext}

\usepackage{tabularx}

\usepackage{tikz}
\usepackage{url}
\usepackage{hyperref}
\usepackage{authblk}

\theoremstyle{plain}
\newtheorem{lemma}{Lemma}[section]

\newtheorem{theorem}[lemma]{Theorem}

\theoremstyle{definition}
\newtheorem{definition}[lemma]{Definition}
\theoremstyle{remark}

\newtheorem*{ulemma}{Lemma}

\usepackage{biblatex}
\title{Understanding the Relationship Between Core Constraints and Core-Selecting Payment Rules in Combinatorial Auctions}

\author[1]{Robin Fritsch}
\author[2]{Younjoo Lee}
\author[3]{Adrian Meier}
\author[4]{Ye Wang}
\author[5]{Roger Wattenhofer}
\affil[1]{rfritsch@ethz.ch\\ETH Zürich}
\affil[2]{youlee@student.ethz.ch\\Seoul National University}
\affil[3]{meiera@student.ethz.ch\\ETH Zürich}
\affil[4]{wangye@ethz.ch\\University of Macau}
\affil[5]{wattenhofer@ethz.ch\\ETH Zürich}

\begin{document}

\maketitle

\begin{abstract}
Combinatorial auctions (CAs) allow bidders to express complex preferences for bundles of goods being auctioned.
However, the behavior of bidders under different payment rules is often unclear.
In this paper, we aim to understand how core constraints interact with different core-selecting payment rules.
In particular, we examine the natural and desirable non-decreasing property of payment rules, which states that bidders cannot decrease their payments by increasing their bids.
Previous work showed that, in general, the widely used VCG-nearest payment rule violates the non-decreasing property in single-minded CAs.
%We provide a significantly smaller and simpler example of this happening in a single-minded CA.
We prove that under a single effective core constraint, the VCG-nearest payment rule is non-decreasing.
In order to determine in which auctions single effective core constraints occur, we introduce a conflict graph representation of single-minded CAs and find sufficient conditions for the single effective core constraint in CAs.
Finally, we study the consequences on the behavior of the bidders and show that no over-bidding exists in any Nash equilibrium for non-decreasing core-selecting payment rules.
\end{abstract}

\section{Introduction}

Combinatorial Auctions (CAs) are widely used to sell multiple goods with unknown value at competitive market prices~\cite{rassenti1982combinatorial}.
CAs permit bidders to fully express their preferences by allowing them to bid on item \textit{bundles} instead of being limited to bidding on individual items.
A CA consists of an allocation algorithm that chooses the winning bidders, and a payment function that determines the winner's payments.
%, which are both evaluated on the placed bids
%Combinatorial Auctions (CAs) allow bidders to bid on sets of items instead of being limited to placing bids on single items~\cite{rassenti1982combinatorial}. 
CAs are popular, sometimes with a total turnover in the billions of US dollars~\cite{ausubel2017practical}.
Often auction designers want an auction to be truthful, in the sense that all bidders are incentivized to reveal their true value.

%Therefore, it is important to understand how bidders are incentivized to deviate from their truthful bids to obtain benefits under different payment rules.

\begin{table}[h]
    \centering
\begin{tabular}{ | >{\centering}m{3cm} | >{\centering}m{1.25cm}| >{\centering}m{1.25cm} | >{\centering\arraybackslash}m{1.2cm} | }
 \hline
  & Local Bidder 1 & Local Bidder 2 & Global Bidder \\
 \hline
 Bundle & $\{A\}$ & $\{B\}$ & $\{A,B\}$ \\
 \hline
% Bid & $b_1=6$ & $b_2=7$ & $b_G=9$ \\
 Bid & $6$ & $7$ & $9$ \\
 \hline
 Allocation & $\{A\}$ & $\{B\}$ & $\{\}$ \\
 \hline
 First-price payment & 6 & 7 & 0 \\
 \hline
% VCG payment & $p_1^{V}=2$ & $p_2^{V}=3$ & $P_G^{V}=0$ \\
 VCG payment & $2$ & $3$ & $0$ \\
 \hline
% VN payment & $p_1^{VN}=4$ & $p_2^{VN}=5$ & $P_G^{VN}=0$ \\
 VN payment & $4$ & $5$ & $0$ \\
 \hline
\end{tabular}
    \caption{An example of an auction with 3 bidders and 2 items. The two local bidders win the auction because they bid $6+7 = 13$, whereas the global bidder only bids $9$.}
    \label{tbl:LLG}
\end{table}

%Simply, using the bids as the payments, i.e., a first-price payment rule, incentivizes the bidders to bid less than their true private value.

Consider the example in Table \ref{tbl:LLG}. This is a so-called Local-Local-Global (LLG) auction; two bidders are local in the sense that they are only interested in one good each, while the global bidder wants to buy all goods.
If the payment scheme is the first-price payment, the winning local bidders would need to pay $6+7 = 13$. They would have been better off by lying, for instance, by bidding a total amount of $10$ only.

The well-known Vickrey-Clarke-Groves (VCG) payment scheme~\cite{vickrey1961counterspeculation,clarke1971multipart,groves1973incentives} is the unique payment function to guarantee being truthful under the optimal welfare allocation. In our example, the VCG payments are much lower. Indeed, VCG payments are often not plausible in practice because of too low payments~\cite{ausubel2006lovely}.
In our example, the VCG payments of $2+3 = 5$ are less valuable than the bid of the global bidder. Therefore, the global bidder and the seller should ignore the VCG mechanism and make a direct deal.

%However, in general, studying LLG auctions cannot fully understand how bidders bid in CAs. 
Core-selecting payment rules, in particular the VCG-nearest (VN) payment~\cite{day2012quadratic}, have been introduced to improve the situation and to guarantee the seller a reasonable revenue~\cite{day2007fair}. 
The VN payment rule selects the closest point to the VCG payments in the \textit{core}, where
the core is the set of payments, for which no coalition is willing to pay more than the winners~\cite{day2008core} (see Figure \ref{fig:core}).

\begin{figure}[ht]
    \centering
    \includegraphics[width=0.8\textwidth]{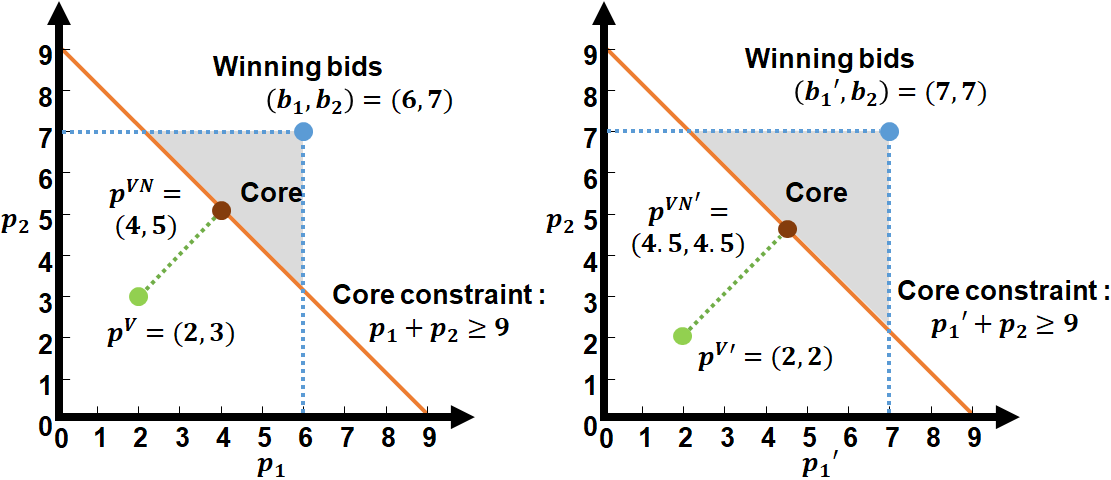}
    \caption{Left: payment space of winning bidders in Table \ref{tbl:LLG}. The green point $p^V$ is the VCG payment point, the red point $P^{VN}$ is the VN payment point, the orange line is the core constraint on payments of local bidders 1 and 2, and the gray triangle is the core given by core constraints. Right: If bidder 1 increases their bid from 6 to 7, their payment increases as well, from 4 to 4.5.}
    \label{fig:core}
\end{figure}

In this paper, we study payment rules for welfare maximizing known single-minded CAs in which each bidder is interested in a single known bundle.
The profile of desired bundles together with the profile of bids define a number of linear constraints (the core constraints) which form a polytope (the core).

However, also core-selecting payments such as the VN payment are not perfect.
It has been shown that bidders can sometimes decrease their payments by announcing \textit{higher-than-truthful} bids under the VN payment rule. Examples which show such overbidding behavior already need a non-trivial amount of goods and bidders~\cite{bosshard2018non}.

In this paper, we study the limitations of VN payments. How complicated can CAs get such that bidders cannot profit from higher-than-truthful bids when using VN payments?
What is the relation between different core constraints and core-selecting payment rules?

%This paper aims to understand how different core constraints interact with payment rules and how they influence the incentives of the bidders.

In particular, we study when the non-decreasing property holds, which is a natural and desirable property of payment rules.
This property requires that a bidder cannot decrease their payment by increasing their bid.
%In general, the VN payment rule is not non-decreasing even in single-minded CAs~\cite{bosshard2018non}.
We examine for which kind of core constraints VN payments are non-decreasing.
More precisely, we show that the non-decreasing property holds whenever a single effective core constraint exists.

Our second result determines which kinds of auctions are non-decreasing. To do so, we introduce a graph-based representation of CAs.
We construct a conflict graph based on the overlap between the desired bundles of the bidders.
We find sufficient conditions on the conflict graph to have a single effective core constraint. In particular, we show that this is the case if the conflict graph is a complete multipartite graph or if any maximal independent set in the conflict graph has at most two nodes.
Furthermore, we show that for auctions with at most three winners, the VN payment is non-decreasing, without relying on the existence of a single effective core constraint.
%Furthermore, we use our insights to find a smaller and simpler example of overbidding in a single-minded CA which helps better understand when overbidding occurs.

Thirdly, we study the consequences on the incentives of the bidders.
We prove that for a non-decreasing payment rule, over-bidding strategy is always weakly dominated by the truthful bidding in any Nash equilibrium in single-minded CAs. This proves a conjecture made by \cite{bosshard2018non}.

Finally, we also study the non-decreasing property for two other common payment rules: the proxy and the proportional payment. Although these two payment rules are not non-decreasing with multi-minded bidders, we prove the non-decreasing property in the case of single-minded bidders.

%This paper makes the following contributions by studying the relations between core constraints and payment rules in CAs: the proof of non-decreasing property of the VN payment rule under single effective core constraints; sufficient conditions on the conflict graph for single effective core constraints in CAs; and the proof of the non-existence of over-bidding behaviors in any Nash equilibrium under non-decreasing payment rules.%; and the proof of the non-decreasing property of the proxy and the proportional payment rules in single-minded CAs.

\section{Related Work}

The incentives of bidders in CAs with core-selecting payment rules are not  understood well~\cite{goeree2016impossibility}. Day and Milgrom claimed that core-selecting payment rules minimize incentives to misreport~\cite{day2008core}. However, it is not known under which circumstances certain incentive properties, like the non-decreasing property, hold.
The non-decreasing property has been observed for the VN payment rule in LLG auctions~\cite{ausubel2020core}, but does not hold in other single-minded CAs~\cite{bosshard2018non}.
%The relationships between core constraints and core-selecting payment rules are not well-understood in CAs. Payment rules may not hold property under some core constraints but have the property under some other core constraints. For example, an incentive property, the non-decreasing property, has been observed in the VCG-nearest payment rule in LLG auction~\cite{ausubel2020core}, but not in some other single-minded CAs~\cite{bosshard2018non}.
Markakis and Tsikirdis examined two other payment rules, $0$-nearest and $b$-nearest, which select the closest point in the minimum-revenue core to the origin and to the actual bids, respectively~\cite{markakis2019core}.
They prove that these two payment rules satisfy the non-decreasing property in single-minded CAs.
%We add to this by finding a necessary condition when the widely used VCG-nearest is non-decreasing in single-minded CAs.

To understand the performance of core-selecting payment rules in CAs, Day and Raghavan propose a constraint generation to codify the pricing problem concisely~\cite{day2007fair}. Later on, B{\"u}nz et al. proposed a faster algorithm to generate core constraints based on the idea of conflict graphs among participants in the auction~\cite{bunz2015faster}.

Payment properties strongly influence incentive behaviors in CAs. Previous research focused on game-theoretic analysis~\cite{day2007fair} and showed that bidders might deviate from their truthful valuation to \textit{under}-bidding strategies (bid shading) or \textit{over}-bidding strategies, where bidders place a bid lower or higher than their valuation, respectively. Ryuji Sano~\cite{sano2011incentives} proved that the truthful strategy is not dominant in proxy and bidder-optimal core-selecting auctions without a triangular condition. However, whether over-bidding strategies exist in any NE is still an open question.

Previous work has shown that in both full and incomplete information setting, under-bidding strategies always exist in Pure Nash equilibria (PNE) and Bayesian Nash equilibria (BNE) for core-selecting CA. Beck and Ott examined over-bidding strategies in a general full-information setting and proved that every minimum-revenue core-selecting CA has a PNE, which only contain over-bids~\cite{ott2013incentives}.
Although the existence of over-bidding strategies in PNE has been proven, incentives for over-bidding when values are private are not very well understood. In BNE, bidders choose from their action space to respond to others' expected strategies with a common belief about the valuation distribution among all bidders.
One of the few known facts is that bidders might over-bid on a losing bundle to decrease their payment for a winning bundle~\cite{ott2013incentives,bosshard2020computing,bosshard2020cost}.

%We make progress in understanding the incentives in the incomplete information setting by studying the existence of over-bidding strategies in Nash equilibria in single-minded CAs. We show that they do not occur for core-selecting non-decreasing payment rules underlining the importance of the non-decreasing property.

Compared to previous studies our work fills the following three gaps.
First, previous studies have not fully considered how core constraints influence core-selecting payment property. This paper examines how core constraints interact with core-selecting payment rules, which motivates better designs of CA models.
Second, %very little has been researched about the conflict graph presentations of core constraints in CAs. This paper suggests that by understanding the graph property of conflict graph, designers of CAs can better understand the property of core-selecting payment rules in CAs.
since we believe that graph representations are at the heart of understanding the core constraints and core-selecting payment rules, we represent conflicts as a graph.
Finally, the relationship between non-decreasing payment rules and incentive behaviors in CAs has not been studied yet. Our work provides new insight into the existence of over-bidding strategies in Nash equilibria, underlining the importance of the non-decreasing property.

\section{Formal Model}

We study auctions under the assumption that all bidders as well as the auctioneer act independently, rationally, and selfishly. Each bidder aims to maximize personal utility.% by rational decision making based on their respective incomplete knowledge.

\subsection{Combinatorial Auctions}

In a combinatorial auction (CA) a set $M=\{1,\ldots, m\}$ of goods is sold to a set $N=\{1,\ldots,n\}$ of bidders.
In this paper, we consider single-minded CAs (SMCAs) in which every bidder only bids on a single bundle.
Let $k_i\subset M$ be the single bundle that bidder $i$ is bidding for and denote $k = (k_1, \ldots,k_n)$ as the interest profile of the auction. We assume that the interest profile of an auction is known and fixed.
Furthermore, let $v_i\in\mathbb{R}_{\geq 0}$ be the true (private) value of $k_i$ to bidder $i$ and $b_i\in\mathbb{R}_{\geq 0}$ the bid bidder $i$ submits for $k_i$.
The bids of all bidders are summarized in the bid profile $b = (b_1,\ldots,b_n)$.
We denote the bid profile of all bids except bidder $i$'s as $b_{-i}$, and in general, the bid profile of a set $L\subset M$ of bidders as $b_L$.

A CA mechanism $(X,P)$ consists of a winner determination algorithm $X$ and a payment function $P$. The winner determination selects the winning bids while the payment function determines how much each winning bidder must pay.

\subsection{Winner Determination}
The allocation algorithm $X(b)$ returns an efficient allocation $x$, i.e.\ a set of winning bidders who receive their desired bundles.
All other bidders receive nothing.
An allocation is called efficient if it maximizes the reported social welfare which is defined as the sum of all winning bids. We denoted the reported social welfare as $W(b, x)=\sum_{i\in x} b_i$.
This optimization problem is subject to the constraint that every item is contained in at most one winning bundle.

Every bidder intends to maximize their utility which is the difference between their valuation of the bundle they acquire and the payment they make.
The social optimum would be to choose the allocation that maximizes the sum of valuations of the winning bundles. However, since the valuations are private, the auctioneer can only maximize the reported social welfare.

\subsection{Payment Functions}
We assume the payment function satisfies voluntary participation, i.e., no bidder pays more than they bid.
So the payment $p_i$ of bidder $i$ satisfies $p_i\leq b_i$ for every $i\in N$.

The Vickrey-Clarke-Groves (VCG) payment is the unique payment rule which always guarantees truthful behavior of bidders in CAs. We denote bidder $i$'s VCG payment as $p_i^V$.

\begin{definition}[VCG payment]
For an efficient allocation $x=X(b)$, the VCG payment of bidder $i$ is
\begin{equation*}
    p_i^V(b,x) := W(b,X(b_{-i}))-W(b,x_{-i})
\end{equation*}
where $x_{-i} = x \setminus \{i\}$ is the set of all winning bidder except $i$.
Note that $X(b_{-i})$ is an efficient allocation in the auction with all bids except bidder $i$'s bid.
\end{definition}

The VCG payment $p_i^V$ is a measurement of bidder $i$'s contribution to the solution. It represents the difference between the maximum social welfare in an auction without $i$ and the welfare of all winners except $i$ in the original auction.

\begin{definition}[Core-selecting Payment Rule]
\label{def:core}
For an efficient allocation $x=X(b)$, the core is the set of all points $p(b,x)$ which satisfies the following constraint for every subset $L \subseteq N:$
$$\sum_{i \in N \setminus L} p_i(b,x) \geq W(b, X(b_L)) - W(b, x_L) $$
Here, $x_L = x \cap L$ is the set of winning bidders in $L$ under the allocation $x$.
Note that $X(b_L)$ is an efficient allocation in the auction with only the bids of bidders in $L$.

A payment rule is called \emph{core-selecting} if it selects a point within the core.
The \emph{minimum revenue core} forms the set of all points $p(b,x)$ minimizing $\sum_{i\in N}p_i(b,x)$ subject to being in the core.
\end{definition}

The core is described by lower bound constraints on the payments such that no coalition can form a mutually beneficial renegotiation among themselves. Those core constraints impose that any set of winning bidders must pay at least as much as their opponents would be willing to pay to get their items. The VN payment rule selects a payment point in the core closest to the VCG point.

\begin{definition}[VCG-nearest Payment]
The VCG-nearest payment rule (quadratic payment, VN payment) picks the closest point to the VCG payment within the minimum-revenue core with respect to Euclidean distance.
\end{definition}

We also study the proxy payment and the proportional payment, which are both core-selecting.

\begin{definition}[Proxy Payment]
The proxy payment selects the point in the core where the winners of the auction will share the total payment equally. It is defined as the point of the form $p_i(b,x) = min[\alpha, b_i]$ for the minimum $\alpha \geq 0$ such that the point is in the core.
\end{definition}
\begin{definition}[Proportional Payment]
With the proportional payment rule, the winning bidders' payments are given by the point in the core that minimizes the total payment $\sum_{i\in N}p_i(b,x)$ and is of the form $p_i(b,x) := \alpha \cdot b_i$ for some $\alpha\in[0, 1]$.
\end{definition}

\begin{definition}[Non-decreasing Payment Rule]
\label{def:non-decreasing-weak}
For any allocation $x$, let $\mathcal{B}_x$ be the set of bid profiles for which $x$ is efficient. The payment-rule $p$ is non-decreasing if, for any bidder~$i$, any allocation $x$, and bid profiles $b, b' \in \mathcal{B}_x$ with $b'_i \geq b_i$ and $b_{-i}=b'_{-i}$, the following holds:
$$ p_i(b',x) \geq p_i(b,x) $$
\end{definition}

\section{Non-decreasing payment rules and single effective core constraints}

\label{sec:secc_non}

We begin by proving a sufficient condition on the core constraints that guarantees that VN is a non-decreasing payment rule.

For core selecting payment rules, the core constraints bound the payments from below to ensure that no collation has a higher reported price than the winners. However, many of the constraints are redundant since other constraints are more restrictive. For example, consider an LLG auction such as the one shown in Figure \ref{fig:core} in which the local bidders win.
The core constraints on their payment are then
\begin{align}
    p_1 + p_2 &\geq b_G\label{eq:LLG1}\\
    p_1 &\geq b_G-b_2\label{eq:LLG2}\\
    p_2 &\geq b_G-b_1\label{eq:LLG3}
\end{align}
where $b_G$ is the bid of the global bidder.
Of these constraints, \eqref{eq:LLG2} is immediately satisfied, as soon \eqref{eq:LLG1} holds since $p_2\leq b_2$.
The same is true for \eqref{eq:LLG3}.
So \eqref{eq:LLG1} is the only effective constraint. We will formalize this idea in the following.
Note that the constraints \eqref{eq:LLG2} and \eqref{eq:LLG3} discussed above are of the form $p_i\geq p_i^V$.
Such a constraint arises for every winning bidder from the core constraint for $N\setminus L=\{i\}$.
However, we can in general disregard core constraints of the form $p_i\geq p_i^V$ which we call VCG-constraints since we are minimizing the distance between $p$ and $p^V$ and no other constraint forces $p_i< p_i^V$.

\begin{definition}
\label{def:single_eff_core_constr}
Consider an SMCA with a fixed interest profile and a fixed winner allocation.
Intuitively, we say a \textbf{single effective core constraint} (SECC) exists, if the fact that a single core constraint holds implies that all other core constraints are satisfied for all bid profiles. 
More formally, an SECC exists, if the polytope defined by this core constraint together with the voluntary participation constraints exactly equals the core (which is defined by all core constraints).
\end{definition}

\begin{theorem}
\label{thm:single_core_constraint}
The VN payment rule is non-decreasing for SMCAs with a single effective core constraint.
\end{theorem}

\begin{proof}
To prove this theorem we will first compute an explicit formula for the VN payments.
The payments of all losing bidders are 0. For all winning bidders whose payment is not part of the SECC, the VN payment simply equals the VCG payment.
Let $S$ be the set of winners whose payment is part of the SECC. Then we have the following constraints on the VN payments to $S$, where \eqref{eq:constraint1} is the SECC with some lower bound $B$.
\begin{align}
    \sum_{i\in S} p_i^{VN} &\geq B \label{eq:constraint1}\\
    p_i^{VN} &\leq b_i \quad\text{ for } i\in S \label{eq:constraint2}
\end{align}
The quadratic optimization problem to be solved is minimizing the Euclidean distance between $p^{VN}$ and $p^{V}$ under the constraints above.
For the solution of this optimization, the voluntary participation constraint \eqref{eq:constraint2} will be active for some $i$.
Let $A$ be the set of indices for which \eqref{eq:constraint2} is active, i.e. $p_i^{VN}=b_i$ for $i\in A$.

For the remaining $i\in S\setminus A$, we write $p_i^{VN} = p_i^{V}+\delta_i$.
The single effective core constraint \eqref{eq:constraint1} can now be rewritten as
\begin{equation*}
    \sum_{i\in S\setminus A} \delta_i \geq B - \sum_{i\in S\setminus A} p_i^V - \sum_{i\in A} b_i.
\end{equation*}
Minimizing the Euclidean distance between $p^{VN}$ and $p^{V}$ is equivalent to minimizing $\sum_{i\in S\setminus A} \delta_i^2$.
Since we have a lower bound on the sum of the $\delta_i$, the minimum possible value of $\sum_{i\in S\setminus A} \delta_i^2$ is achieved when all $\delta_i$ are equal, i.e.
\begin{equation}\label{eq:delta}
    \delta_i = \delta = \frac{1}{|S\setminus A|}\left( B - \sum_{j\in S\setminus A} p_j^V - \sum_{j\in A} b_j \right)
\end{equation}
for $i\in S\setminus A$.
With that we conclude
\begin{equation}\label{eq:VN}
    p_{i}^{VN} = \left\{
                \begin{array}{ll}
                  b_i &\text{for } i \in A\\
                  p_i^V +\delta &\text{for } i \in S\setminus A.
                \end{array}
              \right.
\end{equation}

Finally, we verify that VN is non-decreasing.
Assume bidder $i$ increases their bid and this does not change the allocation $x$.
If $i$ is a losing bidder in $x$, their VN payment is 0 and can obviously not decrease.
Furthermore, if $i$ is a winning bidder, but $i$'s payment is not part of the SECC, $i$'s VN payment will equal their VCG payment which does not change as it only depends on the other bids.
From now on, we assume that bidder $i$ is a winning bidder whose payment is part of the SECC, i.e.\ $i\in S$.

Consider how the quadratic optimization problem changes when increasing bidder $i$'s bid.
One constraint and the point $p^{V}$ move continuously with this change.
So clearly the solution, i.e.\ $p^{VN}$, also moves continuously.
During this move some of the constraints \eqref{eq:constraint2} will become active or inactive.
We call the moments when this happens \emph{switches} and examine the \emph{steps} between two consecutive switches.

As $p^{VN}$ changes continuously around switches, equation \eqref{eq:VN} will yield the same result at the switch, no matter if we consider the switching constraint to be active or not.
So for every single step we can assume that the set of active constraints is the same at the beginning and the end of the step.
If suffices to show that bidder $i$'s payment does not decrease in every step between two switches.

Assume bidder $i$'s bid increases from $b_i$ to $b_i'$ in a certain step and let $b = (b_i,b_{-i})$ and $b' = (b_i',b_{-i})$ denote the corresponding bid profiles.
We distinguish two case based on if $i$ is in the set of active constraints in this step or not.
If $i$'s constraint is active, i.e.\ $i\in A$, we have $p_i^{VN}(b,x) = b_i$ and $p_i^{VN}(b',x) = b_i'$ in \eqref{eq:VN}.
Then the voluntary participation constraint implies
\begin{equation*}
p_i^{VN}(b,x)= b_i\leq b_i' = p_i^{VN}(b',x).
\end{equation*}
Otherwise, for $i\notin A$, we have $p_i^{VN}(b,x) = p_i^V +\delta$ and $p_i^{VN}(b',x) = p_i^V +\delta'$ where $\delta'$ is the term in \eqref{eq:delta} for the bidding profile $b'$ with the increased bid.
Then it remains to argue that $\delta\leq \delta'$.
This is true since neither $B$ nor $|S\setminus A|$ in \eqref{eq:delta} change. The sum $\sum_{j\in A} b_j$ also stays the same since $i\notin A$. Furthermore, $\sum_{j\in S\setminus A} p_j^V$ decreases or stays the same because the VCG payments of all other bidders decrease or stay the same when a winning bidder increases their bid.
\end{proof}

So the existence of an SECC is a sufficient condition for the VN payment to be non-decreasing.
It is however, not a necessary condition as the example in Section \ref{app:non_decreasing_not_necessary} in the appendix shows.

\section{Graph Representation of Auction Classes}
\label{sec:graphs}

In the following, we examine for which auction classes there is guaranteed to exist only a single effective core constraint.
To this end, we consider a representation of the auction classes as graphs.
More precisely, we construct a conflict graph from the interest profile of an auction which represents the overlap between the bundles as follows.
For an interest profile $k=(k_1,\ldots,k_n)$ of an SMCA, consider the graph $G=(V,E)$, where $V=\{k_1,\ldots,k_n\}$, i.e.\ each node represents a bidder.
Two nodes are connected by an edge if and only if the corresponding bundles intersect in at least one item. Two simple examples are shown in Figure \ref{fig:auction_graph}.

\begin{figure}[h]
\centering
\begin{tikzpicture}
\draw[fill=black] (0,0) circle (2pt);
\draw[fill=black] (-1,0) circle (2pt);
\draw[fill=black] (1,0) circle (2pt);
\draw[fill=black] (0,-1) circle (2pt);
\draw[fill=black] (0,1) circle (2pt); 
\draw[thick] (1,0) -- (0,0) -- (-1,0);
\draw[thick] (0,1) -- (0,0) -- (0,-1);

\draw[fill=black] (3,0) circle (2pt);
\draw[fill=black] (4,1) circle (2pt);
\draw[fill=black] (5,0) circle (2pt);
\draw[fill=black] (4,-1) circle (2pt);
\draw[thick] (3,0) -- (4,1) -- (5,0) -- (4,-1) -- (3,0);
\end{tikzpicture}
\caption{
Two examples of conflict graphs. The left one corresponds to the interest profile $(\{A\},\{B\},\{C\},\{D\},\{A,B,C,D\})$, the right one to $(\{A,B\}, \{B,C\}, \{C,D\}, \{D,A\})$.
} 
\label{fig:auction_graph}
\end{figure}
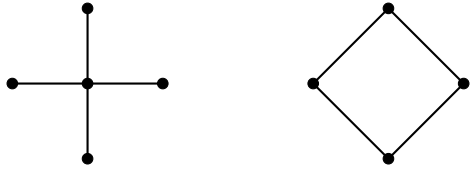
Every set of winners of the auction corresponds to a maximal independent set (MIS) in the graph.

Note that every graph with $n$ nodes is the conflict graph of an SMCA with $n$ bidders, i.e.\ the mapping is surjective: Given a graph, we associate a distinct item with every edge. For every node we then choose the bundle containing all items of adjacent edges.
While different interest profiles are mapped to the same conflict graph, auctions with the same conflict graph lead to equivalent core constraints.

\begin{lemma}
Interest profiles with the same conflict graph have equivalent core constraints (up to renaming the bidders) for all possible bid profiles.
\end{lemma}

\begin{proof}
For two interest profiles with isomorphic conflict graphs, let us renumber the bidders in one profile such that the isomorphism maps the $i$-th bidder in one graph to the $i$-th bidder in the other graph for all $i\in\{1,\ldots,n\}$.
Remember Definition \ref{def:core} of the core constraints:
$$\sum_{i \in N \setminus L} p_i(b,x) \geq W(b, X(b_L)) - W(b, x_L)$$
Note that for every $L$, the sets $x_L$ and $X(b_L)$ depend only on the conflict graphs and the bid profile. So the same is true for the whole right side on the inequality.
\end{proof}

Only by looking at the conflict graph, we can tell by the following sufficient conditions if an SECC exists.

\begin{lemma}\label{lemma:multipartite}
Every auction whose conflict graph is a complete multipartite graph has a single effective core constraint.
\end{lemma}

\begin{proof}
If the conflict graph of an auction is a complete multipartite graph,  the bidders can be grouped into $k$ bidder groups $B_1, \ldots, B_k$, where no edge between two bidders within the same group exists, but any two bidders in different groups are connected by an edge.

We argue that the winning set must be one of these bidder groups.
A winning set clearly can not contain bidders from different groups since their bundles overlap. Moreover, if the winning set is only a subset of a bidder group, the current winner allocation does not maximize reported social welfare, since the rest of the group can simply be added to the winners.

Let $B_w$ be the winning bidder group.
We now argue that only a single effective core constraint exists.
For any subset $L\subseteq N$, we have the core constraint
\begin{equation}\label{eq:core_constraint}
    \sum_{i \in N \setminus L} p_i(b,x) \geq W(b, X(b_L)) - W(b, x_L).
\end{equation}
First, note that the core constraint is not effective if $N\setminus L$ contains a losing bidder. Adding this losing bidder to $L$ does not change the left-hand side (LHS) of \eqref{eq:core_constraint} since this losing bidder's payment must be 0. On the other hand, the right-hand side (RHS) will not decrease since $W(b,x_L)$ does not change. Hence, the new constraint covers the previous one.

So we only need to consider the core constraints with $N\setminus L \subset B_w$. Choose $L'$ such that $B_w\setminus L' = N\setminus L$.
Furthermore, let $B_o$ be the winning bidder group in the auction with only the bidders $N\setminus B_w$.
The term $W(b, X(b_L))$ on the RHS equals either $\sum_{i\in L'}b_i$ or $\sum_{i\in B_o} b_i$.
If the former is true, the RHS is 0 and the constraint is clearly not effective.
In the latter case, the constraint is of the form
\begin{equation*}
    \sum_{i\in B_w\setminus L'} p_i(b,x) \geq \sum_{i\in B_o} b_i - \sum_{i\in L'}b_i.
\end{equation*}
Because of $p_i(b,x)\leq b_i$, any such constraints is covered by the constraint
\begin{equation*}
    \sum_{i\in B_w} p_i(b,x) \geq \sum_{i\in B_o} b_i
\end{equation*}
which is therefore the single effective core constraint.

\end{proof}

Note, both graphs in Figure \ref{fig:auction_graph} are complete bipartite meaning an SECC exists for any auctions with such a conflict graph.
Another sufficient condition for the existence of an SECC is the following.

\begin{lemma}
If every MIS in the conflict graph contains at most 2 nodes, the auction has a single effective core constraint.
\label{lemma:mis2}
\end{lemma}
\begin{proof}
As seen in the previous proof, we only need to consider core constraints where $N \setminus L$ contains only winning bidders. So we get a constraint of the form $p_i + p_j\geq B$, and one each for $p_i$ and $p_j$. These are either $p_i\geq 0$ or $p_i\geq B-b_j$ (and similarly for $p_j$). Because $p_i\leq b_i$ and $p_j\leq b_j$, $p_i + p_j\geq B$ is the only effective core constraint.
\end{proof}

The Lemmas \ref{lemma:multipartite} and \ref{lemma:mis2} show two sufficient conditions for the existence of a single effective core constraint.
They are however not necessary.
This is illustrated by the example shown in Section \ref{app:secc_not_necessary} of the appendix.
While its conflict graph has a MIS of size larger than 2 and is not a complete multipartite graph, we prove that it has a SECC.

So looking at the conflict graph can tell us when the auction has a SECC and consequently, if certain payment rules are non-decreasing for this auction.
On the other hand, by understanding induced subgraphs of the conflict graph, we can also determine that the non-decreasing property of a payments rule is violated for this auction.

\begin{lemma}
Consider two interest profiles $k$ and $k'$ with corresponding conflict graphs $G$ and $G'$.
If $G'$ is an induced subgraph of $G$ and a payment rule is not non-decreasing for $k'$, then the payment rule is also not non-decreasing for $k$.
\end{lemma}

\begin{proof}
According to Definition \ref{def:non-decreasing-weak}, a payment rule not being non-decreasing for $k'$ means there exists an allocation $x$ and bid profiles $b$ and $b'$ with $b_i'\geq b_i$ and $b_{-i}=b_{-i}'$ such that $p_i(b',x)<p_i(b,x)$.
By simply choosing zero (or arbitrarily small) bids for all bidders in $G\setminus G'$, we also find two bid profiles with the same property for $k$.
\end{proof}

Bosshard et al.\ proved that the VN payment violates the non-decreasing property by proposing an interest profile and corresponding bids \cite{bosshard2018non}.
Hence, VN is also not non-decreasing for any auction that contains the graph of this example as an induced subgraph.
This principle motivates the search for minimal examples of overbidding, as well as proving further sufficient conditions for when overbidding does not occur.
In the following, we show a sufficient condition for the non-decreasing property, without relying on the existence of a single effective core constraint.

\begin{theorem}\label{thm:mis3}
The VN-payment rule is non-decreasing for all auctions that have an interest profile for which every winner allocation contains at most three winners.
\end{theorem}

\begin{proof}
The case that the auction is won by two bidders is already treated in Lemma \ref{lemma:mis2}.
Assume three bidders win the auction, and without loss of generality, let the winners be bidders 1, 2 and 3.
Then the core constraints are
\begin{align}
    p_1^{VN} + p_2^{VN} + p_3^{VN} &\geq W(b,X(b_{N\setminus\{1,2,3\}}))\label{eq:3winner_cc1}\\
    p_1^{VN} + p_2^{VN} &\geq W(b,X(b_{N\setminus\{1,2\}})) - b_3\label{eq:3winner_cc2}\\
    p_2^{VN} + p_3^{VN} &\geq W(b,X(b_{N\setminus\{2,3\}})) - b_1\label{eq:3winner_cc3}\\
    p_1^{VN} + p_3^{VN} &\geq W(b,X(b_{N\setminus\{1,3\}})) - b_2\label{eq:3winner_cc4}.
\end{align}
Remember, that we can ignore core constraints of the form $p_i^{VN}\geq p_i^V$ (VCG-constraints).
Furthermore, assume without loss of generality that bidder 3 increases their bid.

Let $M$ be the minimum revenue determined by the core constraints. There are two possibilities for the minimum revenue core:
First, if the plane described by \eqref{eq:3winner_cc1} is not fully covered by the constraints \eqref{eq:3winner_cc2}, \eqref{eq:3winner_cc3} and \eqref{eq:3winner_cc4}, the minimum revenue is $M=W(b,X(b_{N\setminus\{1,2,3\}}))$.
We further discuss this case in the next paragraph.
The second possibility is that the plane described by \eqref{eq:3winner_cc1} is fully covered by the other constraints, and $M>(b,X(b_{N\setminus\{1,2,3\}}))$.
Then the minimum revenue core is a single point determined by equality holding in \eqref{eq:3winner_cc2}, \eqref{eq:3winner_cc3} and \eqref{eq:3winner_cc4}.
Since the right sides of \eqref{eq:3winner_cc2}, \eqref{eq:3winner_cc3} and \eqref{eq:3winner_cc4} are not larger than the right side of \eqref{eq:3winner_cc1}, all three constraints are needed to fully cover the plain.
In particular, bidder 3 must be part of $X(b_{N\setminus\{1,2\}})$, otherwise \eqref{eq:3winner_cc1} implies \eqref{eq:3winner_cc2}, and the plane is not fully covered.
But this means, that increasing $b_3$ does not change the right sides of \eqref{eq:3winner_cc2}. As the same is true for \eqref{eq:3winner_cc3} and \eqref{eq:3winner_cc4}, increasing $b_3$ does not move the minimum revenue core and thereby the VN payment point.

In the following, we assume that constraint \eqref{eq:3winner_cc1} is active, and $M=W(b,X(b_{N\setminus\{1,2,3\}}))$.
We argue similarly to the proof of Theorem \ref{thm:single_core_constraint}:
All changes in the VN payments are continuous in the change of the bid $b_3$.
At any time, a number of constraints are active, and this set of active constraints changes at certain switches.
To prove, the payment does not decrease overall, it suffices to prove it does not decrease between two switches, when the set of active constraints does not change.
In the following, we distinguish three possible cases.

\textbf{1st case:} Only constraint \eqref{eq:3winner_cc1} is active.
Hence, the VCG payments are
\begin{equation*}
    \left(p_1^V+\frac{M-(p_1^V+p_2^V+p_3^V)}{3}, p_2^V+\frac{M-(p_1^V+p_2^V+p_3^V)}{3}, p_3^V+\frac{M-(p_1^V+p_2^V+p_3^V)}{3}\right).
\end{equation*}
When $b_3$ is increased, the minimum revenue $M=W(b,X(b_{N\setminus\{1,2,3\}}))$ does not change. Furthermore, $p_1^V$ and $p_2^V$ stay the same or decrease.
So bidder 3's payment does not decrease according to the formula above.

\textbf{2nd case:}
The constraints \eqref{eq:3winner_cc1} and \eqref{eq:3winner_cc2} are active.
This implies
\begin{align*}
    p_3^{VN} &= b_3 + M - W(b,X(b_{N\setminus\{1,2\}}))\\
    p_1^{VN} + p_2^{VN} &= W(b,X(b_{N\setminus\{1,2\}})) - b_3.
\end{align*}
As $b_3$ increases, $W(b,X(b_{N\setminus\{1,2\}}))$ can increase by at most as much as $b_3$. Hence, the payment $p_3^{VN}$ will not decrease.

\textbf{3rd case:}
Constraints \eqref{eq:3winner_cc1} and \eqref{eq:3winner_cc3} are active. (Note that the case when constraints \eqref{eq:3winner_cc1} and \eqref{eq:3winner_cc4} are active is equivalent due to symmetry.)
This implies
\begin{align*}
    p_1^{VN} &= b_1 + M - W(b,X(b_{N\setminus\{2,3\}}))\\
    p_2^{VN} + p_3^{VN} &= W(b,X(b_{N\setminus\{2,3\}})) - b_1.
\end{align*}
These equations describe a line on which the VN payment points lies.
Increasing $b_3$ does not change the right side of the equation.
Furthermore, it may decrease $p_2^V$, but does not change $p_3^V$.
Since $p^{VN}$ is the closest point to $p^V$ on the line, this can, if it causes a change, only lead to a decrease of $p_2^{VN}$ and an increase of $p_3^{VN}$.

\end{proof}

\section{Over-bidding for Non-decreasing Payment Rules}

In this section, we show that no over-bidding is profitable for non-decreasing payment rules, which proves a conjecture by \cite{bosshard2018non}.

\subsection{Over-bidding on Winning Bids}

As long as an over-bid does not change the winner allocation compared to the truthful bid, it will not increase the bidders utility. This follows directly from the definition of the non-decreasing payment rules: Increasing a bid will not decrease the payment. However, the allocated value stays the same since the allocation does not change.

Since any efficient allocation remains efficient when increasing a winning bid, over-bidding when the truthful bid is already a winning bid is not profitable.
This implies the following lemma.

\begin{lemma}
\label{lem:overbid_winning_SM}
Consider an SMCA with a core-selecting, non-decreasing payment rule.
If for a bidder $i$ and fixed bids of the other bidders $b_{-i}$, the truthful bid $v_i$ is a winning bid, then overbidding decreases bidder $i$'s utility.
\end{lemma}

\subsection{Over-bidding on Losing Bids}

%However, over-bidding on a losing bid may change the efficient allocation.
A losing bid has zero utility due to voluntary participation, i.e., no bundle is acquired, no value is gained, and the payment is zero. A losing over-bid equally results in zero utility. Thus, an over-bidding strategy that increases the utility must result in winning the auction.

The following lemma shows that it is not possible to gain a positive utility by over-bidding, where the truthful bid is losing.
%For the proof see Section \ref{sec:overbid_losing} of the appendix.

\begin{lemma}
\label{lem:overbid_losing_SM}
Consider an SMCA with a core-selecting, non-decreasing payment rule.
If for a bidder $i$ and fixed bids of the other bidders $b_{-i}$, the truthful bid $v_i$ is a losing bid, then over-bidding is not profitable for bidder $i$.
\end{lemma}

\begin{proof}
Consider a bidder $i$ who loses when bidding their truthful private value $v_i$, but wins with an overbid $b_i^o>v_i$.
We write the truthful and the over-bidding bid profile as $b^v = (v_i, b_{-i})$ and $b^o = (b_i^o,b_{-i})$, respectively.
Furthermore, let $x^v$ and $x^o$ denote the efficient allocations for the bidding profiles $b^v$ and $b^o$, respectively. Note that $i\notin x^v$, but $i\in x^o$.

The fact that bidder $i$ loses with bid $v_i$ implies that $W(b^v,x^o) < W(b^v,x^v)$.
Let $\epsilon = W(b^v, x^v) - W(b^v, x^o)$.
We choose $b_i^o = v_i +\epsilon$ as the smallest overbid, such that $x^o$ is an efficient allocation.
Then $W(b^o,x^o) = W(b^v,x^v)$.
Note that it suffices to consider this overbid since any further increase of the bid beyond this value decreases bidder $i$'s utility according to Lemma \ref{lem:overbid_winning_SM}.

We calculate the VCG payment of bidder $i$ for the bidding profile $b^o$. 
The maximum reported social welfare without bidder $i$ equals $W(b^o, x^v)=W(b^v, x^v)$ since bidder $i$ loses in $x^v$.
Furthermore, the total reported social welfare of $x^o$ excluding $i$ equals $W(b^o,x^o)-b_i^o$.
Therefore,
\begin{equation*}
    p_i^V(b^o,x^o) = W(b^v, x^v) - (W(b^o,x^o)-b_i^o) = b_i^o.
\end{equation*}
As mentioned in Section \ref{sec:secc_non}, a core constraint of the form $p_i(b^o,x^o)\geq p_i^V(b^o,x^o)$ exists for every bidder.
But this means that bidder $i$'s payment is at least $b_i^o>v_i$ resulting in a negative utility for bidder $i$.
\end{proof}

Together Lemmas \ref{lem:overbid_winning_SM} and \ref{lem:overbid_losing_SM} imply the desired result.

\begin{theorem}
\label{thm:non-decSCMA}
In an SMCA with a core-selecting, non-decreasing payment rule, over-bidding strategy is always weakly dominated by truthful bidding in any Nash equilibrium.
\end{theorem}

\section{Non-decreasing Property of Other Core-Selecting Payment Rules}

We showed in the previous section, that the non-decreasing property ensures that bidders can not profit from overbidding. However, until now, we have a limited understanding of which payment rules are non-decreasing. In this section, we investigate two common core-selecting payment rules, namely the proxy and the proportional payments, and examine their non-decreasing properties.

Bosshard et al. argued that the proxy and the proportional payments are non-decreasing~\cite{bosshard2018non}. However, a counterexample with two bidders and two items $A$ and $B$ exists: Assume $b_1(A)=12$, $b_1(AB)=18$ and $b_2(B)=9$. Under the proportional rule we have that $p_1 = 8$ and $p_2 = 6$. However, if we increase $b_1(A)$ to 15, we get $p_1 = 5$ and $p_2 = 3$. This counterexample occurs because the core constraint for bidder 2 has been lowered by the increase of $b_1(A)$.
%We correct the proof of Proposition 2 in~\cite{bosshard2018non} by proving the following lemma in Section \ref{app:other_payment_rules} of the appendix.
In this section, we correct the proof of Proposition 2 in~\cite{bosshard2018non}.

\begin{lemma}
\label{claim:prop_prox_non_dec}
The proxy payment function and the proportional payment function are non-decreasing for SMCA.
\end{lemma}

\begin{proof}

We prove the statement for the proportional payment rule, the argument works analogously for proxy payment rule. 

First, we show that for SMCAs with continuous payment function, there exists a core constraint which is active, and that is not affected by increasing bids for any winning bidders.

Consider a winning bidder $i$ increasing their bid by an arbitrarily small amount $\epsilon$ from $b_i$ to $b_i'=b_i+\epsilon$. Because the bids are continuous, there exists a core constraint which is active for both $p$ and $p'$.
Let us denote this core constraint with $CC$, which is generated by the set of bidders $L$. If $i\notin L$, then increasing $b_i$ to $b_i'$ does not change $CC$, which is sufficient for the non-decreasing proof in \cite{bosshard2018non}. Otherwise, we have $i\in L$.

Consider the efficient allocation $X(b_L)$ in the auction with only the bids of bidders in $L$. Since $b_i<b_i'$, there are three scenarios:

\textbf{1st case:} $X(b_L)_i=0$ and $X(b_L')_i=0$. Then $CC$ decreases by $b_i'-b_i$, which is denoted as $CC'$. However, there exists another core constraint $CC_{-i}$ which is generated by the set of $L\setminus i$. Increasing $b_i$ to $b_i'$ does not change $CC_{-i}$. Moreover, constraint
$CC_{-i}$ which is $$\sum_{j \in N \setminus L} p_j(b,x) + p_i \geq W(b, X(b_L)) - (W(b, x_L) - b_i)$$ covers
$CC$: $$\sum_{j \in N \setminus L} p_j(b,x) \geq W(b, X(b_L)) - W(b, x_L).$$
Therefore, $p$ and $p'$ are all on $CC_{-i}$ that is not affected by increasing bids from $b_i$ to $b_i'$.

\textbf{2nd case:} $X(b_L)_i=1$ and $X(b_L')_i=1$. Then $CC$ is not affected by increasing bids from $b_i$ to $b_i'$.

\textbf{3rd case:} $X(b_L)_i=0 = 0$ and $X(b_L')_i = 1$. Then there is a switching point $b_i^*$, such that $i$ is in one of the efficient allocations with bid $b_L^*$, but $i$ is not in the efficient allocations with bids $b_L^-$ for any $b_i^-<b_i^*$.
Then we can first use the same argument as in the first scenario to show that $p$ and $p^*$ are all on the $CC_{-i}$ that is not affected by increasing bids from $b_i$ to $b_i^*$. Then, we use the same argument as in the second scenario to show that $CC_{-i}$ is not affected by increasing bids from $b_i^*$ to $b_i'$.

Therefore, for SMCAs, there exists a core constraint that is active, and is not affected by increasing bids for any winning bidders. Let $\hat{p}$ be the unique point on the line of payments proportional according to $b'$ with $\hat{p}_i=p_i$. Because $b_i<b_i'$, we have $\hat{p}_j<p_j, \forall j\neq i$. Therefore, $\hat{p}$ lies weakly below the core constraint $CC$, thus $p_i<p_i'$ and the proportional payment function are non-decreasing for SMCA.

\end{proof}

\section{Conclusion}
In this paper, we study the relationship between payment rules and core constraints in CAs. We show how core constraints interact with an incentive property of payment rules in SMCAs, more precisely, that a single effective core constraint results in the non-decreasing property of the VN payment rule.
Additionally, we introduce a conflict graph representation of SMCAs and prove sufficient conditions on it for the existence of a single effective core constraint.
%This paper also introduces a conflict graph representation of SMCAs. Moreover, it finds that if every MIS contains at most two nodes, or the conflict graph is complete multipartite, then a single effective core constraint exists for the SMCA.
Furthermore, we examine the conjecture that over-bidding is never favorable in any Nash equilibrium of CAs with non-decreasing payment rules.

% bibliography
\printbibliography

@article{ausubel2020core,
  title={Core-selecting auctions with incomplete information},
  author={Ausubel, Lawrence M. and Baranov, Oleg},
  journal={International Journal of Game Theory},
  volume={49},
  number={1},
  pages={251--273},
  year={2020},
  publisher={Springer}
}

@article{day2012quadratic,
  title={Quadratic core-selecting payment rules for combinatorial auctions},
  author={Day, Robert W and Cramton, Peter},
  journal={Operations Research},
  volume={60},
  number={3},
  pages={588--603},
  year={2012},
  publisher={INFORMS}
}

@inproceedings{bosshard2018non,
  title={Non-decreasing Payment Rules for Combinatorial Auctions.},
  author={Bosshard, Vitor and Wang, Ye and Seuken, Sven},
  booktitle={IJCAI},
  pages={105--113},
  year={2018}
}

@article{ausubel2006lovely,
  title={The lovely but lonely Vickrey auction},
  author={Ausubel, Lawrence M. and Milgrom, Paul and others},
  journal={Combinatorial auctions},
  volume={17},
  pages={22--26},
  year={2006},
  URL={https://www.researchgate.net/profile/Paul_Milgrom/publication/247926036_The_Lovely_but_Lonely_Vickrey_Auction/links/54bdcfe10cf27c8f2814ce6e/The-Lovely-but-Lonely-Vickrey-Auction.pdf}
}

@article{day2008core,
  title={Core-selecting package auctions},
  author={Day, Robert and Milgrom, Paul},
  journal={international Journal of game Theory},
  volume={36},
  number={3-4},
  pages={393--407},
  year={2008},
  publisher={Springer},
  URL={https://link.springer.com/article/10.1007/s00182-007-0100-7}
}

@article{ott2013incentives,
  title={Incentives for overbidding in minimum-revenue core-selecting auctions},
  author={Ott, Marion and Beck, Marissa},
  year={2013},
  publisher={Kiel und Hamburg: ZBW-Deutsche Zentralbibliothek f{\"u}r~…},
  URL={https://www.econstor.eu/handle/10419/79946}
}

@article{day2007fair,
  title={Fair payments for efficient allocations in public sector combinatorial auctions},
  author={Day, Robert W and Raghavan, Subramanian},
  journal={Management science},
  volume={53},
  number={9},
  pages={1389--1406},
  year={2007},
  publisher={INFORMS},
  URL={https://pubsonline.informs.org/doi/pdf/10.1287/mnsc.1060.0662}
}

@article{bosshard2020computing,
  title={Computing bayes-nash equilibria in combinatorial auctions with verification},
  author={Bosshard, Vitor and B{\"u}nz, Benedikt and Lubin, Benjamin and Seuken, Sven},
  journal={Journal of Artificial Intelligence Research},
  volume={69},
  pages={531--570},
  year={2020}
}

@inproceedings{markakis2019core,
  title={On Core-Selecting and Core-Competitive Mechanisms for Binary Single-Parameter Auctions},
  author={Markakis, Evangelos and Tsikiridis, Artem},
  booktitle={International Conference on Web and Internet Economics},
  pages={271--285},
  year={2019},
  organization={Springer}
}

@article{ausubel2017practical,
author = {Ausubel, Lawrence M. and Baranov, Oleg},
title = {A Practical Guide to the Combinatorial Clock Auction},
journal = {The Economic Journal},
volume = {127},
number = {605},
pages = {F334-F350},
doi = {https://doi.org/10.1111/ecoj.12404},
%url = {https://onlinelibrary.wiley.com/doi/abs/10.1111/ecoj.12404},
%eprint = {https://onlinelibrary.wiley.com/doi/pdf/10.1111/ecoj.12404},
year = {2017}
}

@article{goeree2016impossibility,
  title={On the impossibility of core-selecting auctions},
  author={Goeree, Jacob K and Lien, Yuanchuan},
  journal={Theoretical economics},
  volume={11},
  number={1},
  pages={41--52},
  year={2016},
  publisher={Wiley Online Library}
}

@article{bosshard2020cost,
  title={The cost of simple bidding in combinatorial auctions},
  author={Bosshard, Vitor and Seuken, Sven},
  journal={arXiv:2011.12237},
  year={2020}
}

@inproceedings{bunz2015faster,
author = {B\"{u}nz, Benedikt and Seuken, Sven and Lubin, Benjamin},
title = {A Faster Core Constraint Generation Algorithm for Combinatorial Auctions},
year = {2015},
isbn = {0262511290},
publisher = {AAAI Press},
pages = {827–834},
numpages = {8},
location = {Austin, Texas},
series = {AAAI'15}
}

@article{clarke1971multipart,
  title={Multipart pricing of public goods},
  author={Clarke, Edward H},
  journal={Public choice},
  volume={11},
  number={1},
  pages={17--33},
  year={1971},
  publisher={Springer}
}

@article{groves1973incentives,
  title={Incentives in teams},
  author={Groves, Theodore},
  journal={Econometrica: Journal of the Econometric Society},
  pages={617--631},
  year={1973},
  publisher={JSTOR}
}

@article{rassenti1982combinatorial,
  title={A combinatorial auction mechanism for airport time slot allocation},
  author={Rassenti, Stephen J and Smith, Vernon L and Bulfin, Robert L},
  journal={The Bell Journal of Economics},
  pages={402--417},
  year={1982},
  publisher={JSTOR}
}

@article{sano2011incentives,
  title={Incentives in core-selecting auctions with single-minded bidders},
  author={Sano, Ryuji},
  journal={Games and Economic Behavior},
  volume={72},
  number={2},
  pages={602--606},
  year={2011},
  publisher={Elsevier}
}

@article{vickrey1961counterspeculation,
  title={Counterspeculation, auctions, and competitive sealed tenders},
  author={Vickrey, William},
  journal={The Journal of finance},
  volume={16},
  number={1},
  pages={8--37},
  year={1961},
  publisher={JSTOR}
}

\newpage
\appendix

\section{Example from Section \ref{sec:secc_non}}\label{app:non_decreasing_not_necessary}

As proved in Theorem \ref{thm:single_core_constraint}, the existence of an SECC is a sufficient condition for the VN payment to be non-decreasing.
It is however, not a necessary condition as the following example shows (see Table \ref{tbl:bull_graph}).

\begin{table}[ht]
\centering
\begin{tabular}{|c|c|c|c|c|c|}
\hline
Bidder       & 1       & 2         & 3         & 4         & 5       \\ \hline
Bundle $k_i$ & $\{A\}$ & $\{B\}$ & $\{C\}$ & $\{A,B\}$ & $\{A,C\}$ \\
\hline
\end{tabular}
\caption{An example of an SMCA with more than a single effective core constraint that is non-decreasing.}
\label{tbl:bull_graph}
\end{table}

If the bidders $b_1$, $b_2$ and $b_3$ win the auction, we have the following two core constraints:
\begin{align*}
    p_1 + p_2 &\geq b_4 \\
    p_1 + p_3 &\geq b_5
\end{align*}
In general, none of the two fully covers the other so there is more than one effective core constraint.
However, this example is actually non-decreasing.
To prove this, we briefly discuss two necessary conditions for when overbidding can occur.
Firstly, increasing bid $b_i$ must decrease the VCG payment $p_j^V$ of another winner. This is clear since increasing $b_i$ will not change $p_i^V$ and cannot increase another winner's VCG payment. Furthermore, if the VCG payment does not change, the VCG nearest clearly will not either.

Secondly, the decrease of $p_j^V$ must move the VCG nearest point, and not only decrease $p_j^{VN}$, but also $p_i^{VN}$.
The VCG nearest point will move along a number of faces of the core.
Each of the faces if defined by a subset of core constraints being tight.
During the movement along at least one of these faces, $p_i^{VN}$ must decrease.

In this example, there are only the two core constraints $p_1 + p_2 \geq b_4$ and $p_1 + p_3 \geq b_5$. So the only face to consider is the line defined by both constraints being tight. This can be parametrized as
\begin{equation*}
\centering
    \begin{pmatrix}0 \\ b_4 \\ b_5\end{pmatrix} + x \begin{pmatrix}1 \\ -1 \\ -1\end{pmatrix}
\end{equation*}
with $x\in\mathbb{R}$.
Note that only the second and third entry of the directional vector have the same sign.
Hence, 2 and 3 are the only potential indices $j$ such that a decrease of $p_j^V$ could cause $p_j^{VN}$ and $p_i^{VN}$ decrease.
However, the VCG payments are
\begin{align*}
    p_1^V &= \max(b_4-b_2, b_5-b_3)\\
    p_2^V &= \max(b_4-b_1, 0)\\
    p_3^V &= \max(b_5-b_1, 0).
\end{align*}
In particular, increasing $b_2$ or $b_3$ will not decrease $p_3^V$ and $p_2^V$, respectively. So no bidder can decrease their payment by increasing their bid.

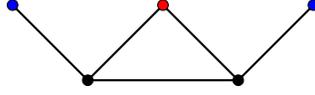
\begin{figure}
\centering
\begin{tikzpicture}
\draw[thick] (-1, 1) -- (0,0) -- (1, 1) -- (2,0) -- (3, 1);
\draw[thick] (0,0) -- (2,0);
\draw[fill=black] (0,0) circle (2pt);
\draw[fill=black] (2,0) circle (2pt);
\draw[fill=blue] (-1, 1) circle (2pt);
\draw[fill=red] (1, 1) circle (2pt);
\draw[fill=blue] (3, 1) circle (2pt);
\end{tikzpicture}
\caption{
The conflict graph of the auction in \autoref{tbl:bull_graph}. The bidders $4,5$ are colored in black. The bidders $2,3$ are colored in blue. The bidder $1$ is colored in red.
} 
\label{fig:bull_graph}
\end{figure}

\section{Example from Section \ref{sec:graphs}}\label{app:secc_not_necessary}

The example in Table \ref{tbl:line5} (with the conflict graph in Figure \ref{fig:line5}) illustrates that while the conditions in Lemmas \ref{lemma:multipartite} and \ref{lemma:mis2} are sufficient, they are not necessary.

\begin{table}[ht]
\centering
\begin{tabular}{|c|c|c|c|c|c|}
\hline
Bidder       & 1       & 2         & 3         & 4         & 5       \\ \hline
Bundle $k_i$ & $\{A,B\}$ & $\{B,C\}$ & $\{C,D\}$ & $\{D,A,E\}$ & $\{E\}$ \\ \hline
\end{tabular}
\caption{An example of an SMCA that has a SECC, where the conflict graph has a MIS of size larger than 2 and is not a complete multipartite graph.}
\label{tbl:line5}
\end{table}

\begin{figure}
\centering
\begin{tikzpicture}
\draw[thick] (-1.5,0) -- (-0.5,1) -- (0.5,0) -- (-0.5,-1) -- (-1.5,0);
\draw[thick] (0.5,0) -- (1.5,0);
\draw[fill=black] (-1.5,0) circle (2pt);
\draw[fill=red] (-0.5,1) circle (2pt);
\draw[fill=black] (0.5,0) circle (2pt);
\draw[fill=red] (-0.5,-1) circle (2pt);
\draw[fill=red] (1.5,0) circle (2pt);
\end{tikzpicture}
\caption{
The conflict graph of the auction in \autoref{tbl:line5}. The MIS with three nodes (bidders $1, 3, 5$) is colored in red.
} 
\label{fig:line5}
\end{figure}
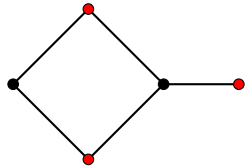

\begin{ulemma}\label{lem:not_necessary}
There exists an SMCA that is not a complete multipartite graph which has an SECC and has an MIS of size larger than 2.
\end{ulemma}

\begin{proof}
Consider the SMCA with the interest profile shown in Table \ref{tbl:line5} (and with the graph representation shown in Figure \ref{fig:line5}).
If two or fewer bidders win, we know from Lemma \ref{lemma:mis2} that an SECC exists. The only possibility for at least 3 bidders winning is if bidders 1, 3 and 5 win.
Then the core constraints on the payments of the winners are the following:

\begin{align}
    p_1+p_3+p_5 &\geq b_2+b_4 \tag{CC1}\\
    p_1+p_5 &\geq \max((b_2+b_4),b_3)-b_3 \tag{CC2}\\
    p_3+p_5 &\geq \max((b_2+b_4),b_1)-b_1 \tag{CC3}\\
    p_1+p_3 &\geq \max((b_2+b_4),(b_2+b_5))-b_5 \tag{CC4}
\end{align}

We first show that CC2 and CC3 are never effective core constraint. If CC2 is non-trivial, then we have $p_1+p_5\geq b_2+b_4-b_3$. Since $p_3\leq b_3$, CC2 is satisfied as soon as the payments of winners satisfies CC1. Therefore, CC2 is never an effective core constraint. As CC3 is symmetric to CC2, so we can use the same argument to show CC3 is also never an effective core constraint.

We now show that the CC1 is always a tight core constraint of any VCG-nearest payment, i.e., $p_1+p_3+p_5= b_2+b_4$. We prove this statement by showing that the following three scenarios are impossible.

The first scenario is that neither CC1 nor CC4 is tight for the VCG-nearest point. In such a scenario, the payment must be equal to the VCG-payment, as those are the remaining constraints on every individual winning bidders payments. Since CC1 is not tight, we have $p_1^V + p_3^V + p_5^V > b_{2} + b_{4}$ which implies

\begin{align*}
    &\max\big((b_{2} + b_{4}), (b_1 + b_3), (b_{2} + b_5)\big)\\
      + &\max\big((b_1 + b_5), (b_{2} + b_{4}), (b_{2} + b_5)\big)\\
      + &\max\big((b_{2} + b_{4}), (b_1 + b_3)\big) > b_{2} + b_{4} + 2(b_{1} + b_{3} + b_{5}).
\end{align*}

If $(b_{2} + b_{4})$ is maximal in any of the three items, we see that the above inequality cannot hold, due to the fact that bidders $1, 3, 5$ winning which means $\max(b_{2} + b_{4}, b_{2} + b_5) \leq b_{1} + b_3 + b_{5}$. Therefore, we obtain the following expression,
    
\begin{equation*}
\begin{split}
    \max\big((b_3 + b_5), (b_{2} + b_5)\big) + \max\big((b_1 + b_5), (b_{2} + b_5)\big)\\
    + (b_1 + b_3)
    > b_{2} + b_{4} + 2(b_{1} + b_{3} + b_{5})
\end{split}
\end{equation*}
    
However, all possible outcomes of the left side are strictly inferior to the right side. Therefore, $p_1^V + p_3^V + p_5^V \leq b_{2} + b_{4}$, which contradicts to our assumption for the first scenario that CC1 is not tight.

The second scenario is that CC4 is tight and CC1 is not tight. Then we have either $p_1+p_3=b_2+b_4-b_5$, or $p_1+p_3=b_2$.

If $p_1+p_3=b_2+b_4-b_5$, then $p_1+p_3+p_5\leq b_2+b_4$ because $p_5\leq b_5$. This implies that CC1 is tight, which contradicts to our assumption of the second scenario.

If $p_1+p_3=b_2$, we have $p_5> b_4$ because CC1 is not tight. Then we know that $p_5$ equals the VCG-payment $p_5^V$, as it is the remaining constraint on bidder $5$'s winning payment. However, $p_5^V> b_4$ implies that,

\begin{align*}
    \max((b_2+b_4), (b_1+b_3))-(b_1+b_3)&>b_4\\
    b_2+b_5&>b_1+b_3+b_5
\end{align*}

This contradicts the assumption that bidders $1, 3, 5$ are winning. Therefore, CC1 is a tight constraint for all VCG-nearest payment in the SMCA.

The next step of the proof is to show that CC4 is always covered by CC1. If $b_2+b_4>b_2+b_5$, because $b_5\geq p_5$, it is trivial that CC4 always satisfies when CC1 satisfies. Otherwise, if $b_2+b_5>b_2+b_4$, we subtract CC4 $p_1+p_3\geq b_2$ from CC1 $p_1+p_3+p_5= b_2+b_4$ (as CC1 is always tight). Then we have $p_5 \leq b_4$, which indicates that if CC1 holds, then CC4 also holds. Therefore, CC1 is the only effective core constraint in the SMCA.
\end{proof}

\end{document}